\documentclass[12pt]{article}
\usepackage[font=small,labelfont=bf]{caption}
\usepackage{amsmath, amsfonts,latexsym,amsthm, amsxtra, amssymb,graphicx,color}
\usepackage{cite}
\usepackage{graphicx}

\usepackage[colorlinks,
            pdfborder=001,
            linkcolor=blue,
            anchorcolor=blue,
            citecolor=blue
            ]{hyperref}

\newtheoremstyle{mythm}{1.5ex plus 1ex minus .2ex}{1.5ex plus 1ex minus .2ex}
    {\it}{}{\bf}{}{0.5em}{}
\theoremstyle{mythm}

\newtheorem{thm}{Theorem}

\newtheorem{ex}{Example}
\numberwithin{thm}{section}
\numberwithin{rem}{section}
\numberwithin{ex}{section}

\def\beq{\begin{equation}}
\def\eeq{\end{equation}}

\usepackage{titlesec}
\titleformat{\section}[hang]{\normalsize\bfseries}{\thetitle\quad}{0.0ex}{}
\titleformat{\subsection}[hang]{\sl\normalsize}{\thetitle\quad}{0.0ex}{}
\titlespacing*{\section} {0pt}{10pt}{4pt}

\usepackage[hang,flushmargin]{footmisc} 

\begin{document}

\begin{center}{\Large\bf  Jackknife empirical likelihood with complex surveys\footnotetext{Supported by}}
\vskip 4mm

{\bf Mengdong Shang, Xia Chen\footnote{Corresponding author. E-mail: xchen80@snnu.edu.cn}}

\small School of Mathematics and Statistics, Shaanxi Normal University, Xi'an 710119, China

\end{center}
\vskip 4mm

\noindent{\bf Abstract}\quad
We propose the so-called jackknife empirical likelihood approach for the survey data of general unequal probability sampling designs, and analyze parameters defined according to U-statistics. We prove theoretically that jackknife pseudo-empirical likelihood ratio statistic is asymptotically distributed as a chi-square random variable, and can be used to construct confidence intervals for complex survey samples. In the process of research, we consider with or without auxiliary information, utilizing design weights or calibration weights. Simulation studies are included to examine that in terms of coverage probability and tail error rates, the jackknife pseudo-empirical likelihood ratio confidence intervals are superior to those based on the normal approximation.


\noindent{\bf 2000 MR Subject Classification}\quad 62F12, 62G05

\section{Introduction}
With the development of society and economy, complex survey has become a powerful approach for collecting data in many fields of scientific investigation. For complex survey, a very significant development in recent years is the use of empirical likelihood method, which does not require any distribution assumptions. Owen (1988,1990) \cite{Owen88, Owen90} first introduced empirical likelihood for independent data, as a way to construct confidence regions for nonparametric and semiparametric inference. It preserves two key properties of parametric likelihood: Wilks' theorem and Bartlett correction, which aid empirical likelihood extend many applications.

In surveys, Hartley and Rao (1968) \cite{Hartley68} first uesd the concept of empirical likelihood under the name "scale-load" approach, and Chen and Qin (1993) \cite{Chen93} first formally applied this method for estimating the population mean under simple random sampling. Chen and Sitter (1999) \cite{Chen99} considered the pseudo-empirical likelihood method of utlizing auxiliary information in complex surveys for general unequal probability sampling design. They proved that if known auxiliary variables are used to estimate the population mean, the proposed method is asymptotically equivalent to the generalized regression estimation method. Wu and Rao (2006) \cite{Wu06} presented pseudo-empirical likelihood ratio confidence intervals for a single parameter under arbitrary sampling design. For the general unistage unequal probability sampling design, Rao and Wu (2010) \cite{Rao10} proved that the pseudo-empirical likelihood method can be used to construct asymptotically valid Bayesian pseudo-empirical likelihood intervals, and is very flexible in the use of auxiliary population information.

In order to extend empirical likelihood to nonlinear statistics, Jing et al.(2009) \cite{Jing09} established the jackknife empirical likelihood for scalar parameters of single-sample and two-sample U-statistic estimation, and proved that the performance is superior to that of the scalar empirical likelihood method proposed by Qin and Zhou (2006) \cite{Qin06}. Li et al.(2016) \cite{Li16} extended the jackknife empirical likelihood method to apply to both vector parameters and nonsmooth estimation equation.

In this article, our primary aim is to construct the jackknife pseudo-empirical likelihood method for the parameters defined by the U-statistic under the general unistage unequal probability sampling design. In Section 2, we respectively discuss the jackknife empirical likelihood method with basic design weights and with calibration weights. In Section 3, we examine the asymptotic properties of jackknife pseudo-empirical log-likelihood ratio functions without or with auxiliary population information. Results of some simulation studies are reported in Section 4. An application to the data set from the 2019 China Household Finance Survey (CHFS) is presented in Section 5. Section 6 contains some conclusion. Proofs are given in Appendix.

\section{Methodology}
For verifying asymptotic theory, we assume that a sequence of finite populations indexed by $v$, is composed of $N_{v}$ units. Exploiting the probability sampling method, we obtain a survey smaple of size $n_{v}$ from population. The population size $N_{v}$ and the sample size $n_{v}$ both incline to $\infty$ as $v \rightarrow \infty $. To simplify the notation, we omit the index $v$ and the limiting process $v \rightarrow \infty $ is replaced by $N \rightarrow \infty $ and $n \rightarrow \infty $.

The population of the $i$th unit is related to the value $y_{i}$ of the variable $y$ and the value $\mathbf{x}_{i}$ of the vector of auxiliary variables $\mathbf{x}$. And we suppose that the population mean vector $\bar{\mathbf{X}} = N^{-1} \sum_{i=1}^{N} \mathbf{x}_{i}$ is known. Let $s$ be the set of $n$ units contained in the sample, and $\pi_{i} = P \left ( i \in s \right ), i=1, \ldots ,N$, be the sample inclusion probabilities.

We use $\theta$ to represent a finite population parameter of interest. Let the $U-$statistic
\beq
T_{n}=T\left ( y_{1},\ldots ,y_{n} \right )=\begin{pmatrix}n\\ m\end{pmatrix}^{-1}\sum_{1\leq i_{1}< \cdots <i_{m}\leq n}h\left ( y_{i_{1}},\ldots,y_{i_{m}} \right ) \nonumber
\eeq
be an unbiased estimator of $\theta$ and $h \left ( \cdot \right)$ be a kernel function. Define the jackknife pseudo-value,
\beq
\hat{V}_{i}=nT_{n}-\left ( n-1 \right )T_{n-1}^{\left ( -i \right )} \nonumber
\eeq
as a new estimator of $\theta$ related to the $i$th unit, with $T_{n-1}^{\left ( -i \right )}=T\left ( y_{1},\ldots ,y_{i-1},y_{i+1},\ldots ,y_{n} \right )$.

\subsection{Jackknife empirical likelihood with basic design weights}
Let us begin by introducing unistage sampling designs with fixed sample size $n$, where $d_{i} = 1/ \pi_{i}, i \in s$ are the basic design weights and $\tilde{d}_{i} \left ( s \right ) = d_{i} / \sum_{i \in s} d_{i}$ are referred to as the normalized design weights for the given sample $s$.

Chen and Sitter (1999) first proposed the pseudo-EL method for complex survey data, and Wu and Rao (2006) defined the profile pseudo-empirical log-likelihood function (PELL) for $\theta$. We now generalize the PELL to the jackknife pseudo-value, which
acquire
\beq \label{JEL1}
l_{JEL} \left ( \theta \right ) = n^{\ast} \sum_{i \in s} \tilde{d}_{i} \left ( s \right ) log \left \{ p_{i} \left ( \theta \right ) \right \}
\eeq
and lend it much boader applicability. The pseudo-MELE of the parameter $\theta$ is given in terms of the Hajek estimator $\hat{V}_{H} = \sum_{i \in s} \hat{p}_{i} \hat{V}_{i} = \sum_{i \in s} \tilde{d}_{i} \left ( s \right ) \hat{V}_{i}$.

With the absence of auxiliary population information in mind, the $n^{\ast}= n / deff_{H}$ is the effective sample size, where the design effect $deff_{H}$ based on the Hajek estimator $\hat{V}_{H}$ is defined as
\beq \label{deffH}
deff_{H} = V_{p} \left ( \hat{V}_{H} \right ) / \left ( S_{V}^{2} / n \right ).
\eeq
The $V_{p} \left ( \cdot \right )$ denotes the variance under the given design and the $S_{V}^{2} / n$ is the variance of $\hat{V}_{H}$ under simple random sampling. For a fixed $\theta$, maximizing (\ref{JEL1}) subject to $p_{i}>0$, \ $\sum_{i \in s}p_{i}=1$ and 
\beq \label{Con1}
\sum_{i \in s}p_{i}\hat{V_{i}}=\theta,
\eeq
leads to $\hat{p}_{i} \left ( \theta \right ) = \tilde{d}_{i} \left ( s \right ) / \left \{ 1 + \lambda \left ( \hat{V}_{i} - \bar{V} \right ) \right \}, \bar{V} = \frac {1} {n} \sum_{i \in s} \hat{V}_{i}$, where the Lagrange multipiler $\lambda$ is the solution to 
\beq \label{Sol1}
\sum_{i \in s} \frac{\tilde{d}_{i} \left ( s \right ) \left ( \hat{V}_{i} - \bar{V} \right )}{1 + \lambda \left ( \hat{V}_{i} - \bar{V} \right )} = 0.
\eeq

We now consider the auxiliary population information by introducing the additional constraint 
\beq \label{Con2}
\sum_{i \in s}p_{i}\mathbf{x}_{i}= \bar{\mathbf{X}}.
\eeq
It follows from Wu and Rao (2006) that the GREG estimator of $\bar{V}$ is given by $\hat{V}_{GR} = \hat{V}_{H} + {B}' \left ( \bar{\mathbf{X}} - \hat{\bar{\mathbf{X}}}_{H} \right )$, where $\hat{\bar{\mathbf{X}}}_{H} = \sum_{i \in s} \tilde{d}_{i} \left ( s \right ) \mathbf{x}_{i}$ is the Hajek estimator and
\beq \label{B}
B = \left \{ \frac{1}{N} \sum_{i=1}^{N} \left ( \mathbf{x}_{i} - \hat{\bar{\mathbf{X}}} \right ) {\left ( \mathbf{x}_{i} - \hat{\bar{\mathbf{X}}} \right ) }' \right \}^{-1} \left \{ \frac{1}{N} \sum_{i=1}^{N} \left ( \mathbf{x}_{i} - \hat{\bar{\mathbf{X}}} \right )\left ( \hat{V} _{i} - \bar{V} \right ) \right \}
\eeq
is the vector of population regression coefficients. In this case, the design effect associated with $\hat{V}_{GR}$ is distinct from (\ref{deffH}), which is defined as
\beq
deff_{GR} = V_{p} \left ( \hat{V}_{GR} \right ) / \left ( S_{r}^{2} / n \right ),\nonumber
\eeq
where $\hat{V}_{GR} = \sum_{i \in s} \tilde{d}_{i} \left ( s \right ) r_{i}$, $r_{i} = \hat{V}_{i} - \bar{V} - {B}' \left ( \mathbf{x}_{i} - \bar{\mathbf{X}} \right )$, $S_{r}^{2} / n$ is the variance of $\hat{V}_{GR}$ under simple random sampling, and $S_{r}^{2} = \left ( N - 1 \right )^{-1} \sum_{i = 1}^{N} r_{i}^{2}$. Accordingly, the effective sample size is given by $n^{\ast}= n / deff_{GR}$. 

Let $u_{i} = {\left ( \hat{V}_{i} - \theta, {\left ( \mathbf{x}_{i} - \bar{\mathbf{X}} \right )}' \right ) }' $. It can be shown that for a fixed $\theta$, maximizing (\ref{JEL1}) subject to $p_{i}>0$, \ $\sum_{i \in s}p_{i}=1$, (\ref{Con1}) and (\ref{Con2}) gives $\tilde{p}_{i} \left ( \theta \right ) = \tilde{d}_{i} \left ( s \right ) / \left \{ 1 + {\bf{\lambda}}' \left ( \theta \right ) u_{i} \right \}$, where the vector-valued Lagrange multipiler $\bf{\lambda} = \bf{\lambda} \left ( \theta \right )$ is the solution to 
\beq \label{Sol2}
\sum_{i \in s} \frac{\tilde{d}_{i} \left ( s \right ) u_{i}}{1 + {\bf{\lambda}}' u_{i}} = 0.
\eeq
In fact, if the zero vector is an interior point of the covex hull of the set $\left \{ u_{i}, i \in s \right \}$, a solution to (\ref{Sol2}) exists and is unique, which was illustrated by Chen, Sitter and Wu (2002). 

\subsection{Jackknife empirical likelihood with calibration weights}
We now turn to the calibration weights $w_{i}$, which released public by statistical agencies. Let $\tilde{w}_{i} \left ( s \right ) = w_{i} / \sum_{i \in s} w_{i}$ be referred to as the normalized calibration weights and $\sum_{i \in s} \tilde{w}_{i} V_{i}$ is calibrated in the sense of $\sum_{i \in s} \tilde{w}_{i}\mathbf{x}_{i} = \mathbf{X}$, where $\mathbf{X} = N\bar{\mathbf{X}}$ is the known population total of the auxiliary vector $\mathbf{x}$.

The following jackknife pseudo-empirical log-likelihood function for calbration weights, which was proposed by Rao and Wu (2010), is given by 
\beq \label{JEL2}
l_{JEL} \left ( \theta \right ) = m \sum_{i \in s} \tilde{w}_{i} \left ( s \right ) log \left ( p_{i} \left ( \theta \right ) \right ),
\eeq
where $m$ is a scale factor to be determined. Here, we choose $m$ as a design consistent estimator of $S_{V}^{2} / V_{p} \left ( \hat{V}_{GR} \right )$. For a fixed $\theta$, maximizing (\ref{JEL2}) subject to $p_{i} > 0$, \ $\sum_{i \in s} p_{i} = 1$ and constraint (\ref{Con1}), leads to $\breve{p}_{i} \left ( \theta \right ) = \tilde{w}_{i} \left ( s \right ) / \left \{ 1 + \lambda \left ( \hat{V}_{i} - \bar{V} \right ) \right \}$, where the Lagrange multipiler $\lambda$ is the solution to 
\beq \label{Sol3}
\sum_{i \in s} \frac{\tilde{w}_{i} \left ( s \right ) \left ( \hat{V}_{i} - \bar{V} \right )}{1 + \lambda \left ( \hat{V}_{i} - \bar{V} \right )} = 0.
\eeq
We note that the calibration equation (\ref{Con2}) is no longer imposing restrictions on $l_{JEL}$.

\section{Main results}
In this section, we verify that under certain regularity conditions, without or with auxiliary population information, the jackknife pseudo-empirical log-likelihood ratio functions with respect to $l_{JEL}$ are both asymptotically distributed as $\chi_{1}^2$. The proof of theorems is given in appendix.

\subsection{Jackknife empirical likelihood without auxiliary population information}
For no auxiliary population information at the estimation stage, maximizing $l_{JEL}$ subject to $p_{i} > 0$ and $\sum_{i \in s}p_{i} = 1$ gives $\hat{p}_{i} = \tilde{d}_{i} \left ( s \right )$, and the $n^{\ast} = n / deff_{H}$. The jackknife pseudo-empirical log-likelihood ratio function related to $\hat{p}_{i} \left ( \theta \right )$ is given by 
\beq
\hat{r}_{JEL} \left ( \theta \right ) = -2 \left \{ l_{JEL} \left ( \hat{p} \left ( \theta \right ) \right )- l_{JEL} \left ( \hat{p} \right ) \right \}. \nonumber
\eeq

Suppose the following regularity conditions hold

\noindent{\bf (C1)}\quad The sampling design $p\left ( s \right )$, the variable $y$ and the kernel function $h \left ( \cdot \right )$ satisfy \\
$max_{i_{1}, \ldots, i_{m} \in s}\left | h \left( y_{i_{1}}, \ldots, y_{i_{m}} \right ) \right |=o_{p}\left ( n^{-\frac{1}{2}} \right )$, where the random order $o_{p}\left ( \cdot  \right )$  is relevant for the sampling design $p\left ( s \right )$. 

\noindent{\bf (C2)}\quad The sampling design $p\left ( s \right )$ satisfies $N^{-1} \sum_{i \in s}d_{i}-1=O_{p} \left (n^{-\frac {1}{2}}\right) $. \\
We research the asymptotic distribution of $\hat{r}_{JEL} \left ( \theta \right )$ as $n \rightarrow \infty $.

Condition (C1) imposes some constraints on the sampling design $p \left ( s \right )$, the kernel function $h \left ( \cdot \right )$ and the finite population $\left \{ y_{1}, \ldots  , y_{N} \right \}$. On the basis of these constraints, it can be derived that $max_{i \in s} \left | nT_{n}-\left ( n-1 \right )T_{n-1}^{\left ( -i \right )} \right | = o_{p}\left ( n^{\frac{1}{2}} \right )$, namely $max_{i \in s} \left | \hat{V}_{i} \right | = o_{p}\left ( n^{\frac{1}{2}} \right )$. Condition (C2) illustrates that $\hat{N} = \sum_{i \in s} d_{i}$ is a $\sqrt{n}$ -consistent estimator of $N$.

\begin{thm}\label{thm1}
Under the conditions {\rm (C1)} and {\rm (C2)}, the jackknife pseudo-jackknife empirical log-likelihood ratio statistic $\hat{r}_{JEL} \left ( \theta \right )$ is asymptotically distributed as $\chi_{1}^{2}$.

\end{thm}

\subsection{Jackknife empirical likelihood with auxiliary population information}
For basic design weights with known auxiliary vectors $\mathbf{x}_{i}$, maximizing $l_{JEL}$ subject to $p_{i} > 0$, $\sum_{i \in s}p_{i} = 1$ and (\ref{Con2}), leads to $\tilde{p}_{i}$ and the $n^{\ast} = n / deff_{GR}$. The jackknife pseudo-empirical log-likelihood ratio function related to $\tilde{p}_{i} \left ( \theta \right )$ is given by 
\beq
\tilde{r}_{JEL} \left ( \theta \right ) = -2 \left \{ l_{JEL} \left ( \tilde{p} \left ( \theta \right ) \right )- l_{JEL} \left ( \tilde{p} \right ) \right \}. \nonumber
\eeq

Suppose the extra regularity condition on the auxiliary vectors $\mathbf{x}_{i}$ hold

\noindent{\bf (C3)}\quad $max_{i \in s}\left \| x_{i}  \right \| = o_{p} \left ( n^{\frac{1}{2}} \right )$, where $\left \| \cdot \right \|$ denotes the $L_{1}$ norm. \\
We can derive the asymptotic distribution of $\tilde{r}_{JEL} \left ( \theta \right )$ as $n \rightarrow \infty $.

\begin{thm}\label{thm2}
Under the conditions {\rm (C1)-(C3)}, the pseudo-jackknife empirical log-likelihood ratio statistic $\tilde{r}_{JEL} \left ( \theta \right )$ is asymptotically distributed as $\chi_{1}^{2}$.
\end{thm}

For calibration weights with known auxiliary vectors $\mathbf{x}_{i}$, due to the same constraints, the $\breve{p}$ and $\breve{r}_{JEL}$ has the identical form as the $\hat{p}$ and $\hat{r}_{JEL}$ except for the different weights. Hence, a result similar to Theorem \ref{thm1} is obtained. 

Suppose the following regularity condition hold

\noindent{\bf (C4)}\quad The sampling design $p\left ( s \right )$ satisfies $N^{-1} \sum_{i \in s}w_{i}-1=O_{p} \left (n^{-\frac {1}{2}}\right) $. \\
We research the asymptotic distribution of $\breve{r}_{JEL} \left ( \theta \right )$ as $n \rightarrow \infty $. Maximizing $l_{JEL}$ subject to $p_{i} > 0$ and $\sum_{i \in s}p_{i} = 1$, leads to $\breve{p}_{i}$ and the $m = S_{V}^{2} / V_{p} \left ( \hat{V}_{GR} \right ) $. The jackknife pseudo-empirical log-likelihood ratio function related to $\breve{p}_{i} \left ( \theta \right )$ is given by 
\beq
\breve{r}_{JEL} \left ( \theta \right ) = -2 \left \{ l_{JEL} \left ( \breve{p} \left ( \theta \right ) \right )- l_{JEL} \left ( \breve{p} \right ) \right \}. \nonumber
\eeq
%

\begin{thm}\label{thm3}
Under the conditions {\rm (C1),(C2) and (C4)}, the pseudo-jackknife empirical log-likelihood ratio statistic $\breve{r}_{JEL} \left ( \theta \right )$ is asymptotically distributed as $\chi_{1}^{2}$.
\end{thm}

\vskip 3mm

\section{Simulation studies}
In this section, under the Rao-Sampford method (Rao 1965, Sampford 1967) of sampling without replacement with inclusion probabilities $\pi_{i}$ by using the jackknife pseudo-empirical likelihood, we conduct simulation studies to examine the performance of coverage probabilities (CP), lower (L) and upper (U) tail error rates, average lengths of confidence intervals (AL) and average lower bound (LB) for the 95\% confidence intervals.

To generate finite populations from model I used by Wu and Rao(2006) 
\beq
y_{i}=\beta _{0}+\beta _{1}x_{i}+\sigma \varepsilon _{i}, \nonumber
\eeq
where $\beta _{0}=\beta _{1}=1$, $x_{i}\sim exp\left ( 1 \right )$ and $\varepsilon _{i}\sim N \left( 0,1 \right )$, population size $N=1000$, sample size $n=100$ and $150$, corresponding to sampling frcations 10\% and 15\% respectively. An appropriate constant number was added to all $x_{i}$ to get ride of extremely small values of $x_{i}$. We use two different values of $\sigma$ to reflect the different correlation between $y$ and $x$: $\rho \left ( y,x \right )=0.3$ and $0.5$. The resulting finite populations remain constant in repeated simulation runs. Our simulations are programmed in R using the algorithms outlined in Wu (2005) and {\it http://www.blackwellpublishing.com/rss}.

We consider two example: the so-called probability weighted moment and the estimate of variance. For the jackknife pseudo-empirical likelihood method, the interval based on no auxiliary variable is denoted by $JEL$, the interval by using the basic design weights $d_{i}$ is denoted by $JEL_{d}$ and the interval by using the calibration weights $w_{i}$ is denoted by $JEL_{w}$. Since the U-statistic has asymptotic normality, we report the results of the interval based on the normal approximation to the usual Z-statistic, which is denoted by $NA$. All the results are based on $B = 1000$ simulation runs.

\begin{ex}\label{ex1}\rm\;\
Let the probability weighted moment, $\theta = E \left \{ yF\left (  y\right ) \right \}$, where $F$ is the distribution function. Then, the sample probability weiighted moment is a $U-$ statistic with the kernel $h= max\left ( x,y \right )/2$.
\end{ex}

Table \ref{tab1} reports the results of example \ref{ex1}, which can be summarized as follows.

(a) The JEL method tail error rates perform more balance than those from method NA, with coverage probabilities closer to the nominal value, but the average length is slightly bigger. At each tail, the method NA leads to smaller L and larger U than the nominal 2.5\% rate. For instance, with n=100 and $\rho \left ( y,x \right ) = 0.3$, L=0.8 and U=7.7 for method NA compared to L=2.7 and U=4.2 for method JEL.

(b) The $JEL_d$ and $JEL_w$ methods intervals both have performance that is similar to that of method NA, but the former have better coverage probability. For instance, when $\rho \left ( y,x \right ) = 0.3$, CP=92.6\% for $JEL_d$ and $JEL_w$ compared to CP=91.5\% for NA with n=100.

(c) The $JEL_d$ and $JEL_w$ methods have comparable lower bound (LB) and the latter has the largest one in all cases, a characteristic which is available in some applications. The smaller lower tail error rate (L) for NA is related to the smaller lower bound (LB).

\begin{table}[htbp]
\centering
\small 
\caption{Covergae probabilties, average lenths of the confidence intervals and tail error rates on the probability weighted moment in Example \ref{ex1} when the norminal level is 0.95.}\vskip -1mm
\label{tab1}
\begin{tabular}{cccccccc}
					\hline
					\makebox[0.1\textwidth][c]{$\rho$}  & \makebox[0.1\textwidth][c]{n} & \makebox[0.1\textwidth][c]{CI} & \makebox[0.1\textwidth][c]{CP(\%)} & \makebox[0.1\textwidth][c]{L} & \makebox[0.1\textwidth][c]{U} & \makebox[0.1\textwidth][c]{AL} & \makebox[0.1\textwidth][c]{LB}\\
					\hline
					0.3	&	100	&	 		 NA		 	 & 91.5 & 0.8 & 7.7 & 0.600 & 3.588\\
							&			&			 JEL		 	 & 93.1 & 2.7 & 4.2 & 0.627 & 3.611\\
							&			&   $JEL_{d}$  	 & 92.6 & 1.5 & 5.9 & 0.598 & 3.616\\
							&			&   $JEL_{w}$  	 & 92.6 & 1.5 & 5.9 & 0.598 & 3.619\\
							&	150	&			 NA		 	 & 91.6 & 1.1 & 7.3 & 0.481 & 3.657\\
							&			&			 JEL		 	 & 93.5 & 1.5 & 5.0 & 0.499 & 3.667\\
							&			&   $JEL_{d}$  	 & 92.4 & 1.3 & 6.3 & 0.480 & 3.675\\
							&			&   $JEL_{w}$  	 & 92.6 & 1.3 & 6.1 & 0.479 & 3.676\\
					0.5	&	100	&	 		 NA		 	 & 90.9 & 0.7 & 8.4 & 0.321 & 3.368\\
							&			&			 JEL		 	 & 93.6 & 2.2 & 4.2 & 0.375 & 3.368\\
							&			&   $JEL_{d}$  	 & 91.9 & 1.4 & 6.7 & 0.319 & 3.382\\
							&			&   $JEL_{w}$  	 & 91.5 & 1.7 & 6.8 & 0.321 & 3.386\\
							&	150	&			 NA		 	 & 90.8 & 1.0 & 8.2 & 0.257 & 3.406\\
							&			&			 JEL		 	 & 94.7 & 1.3 & 4.0 & 0.296 & 3.399\\
							&			&   $JEL_{d}$  	 & 91.9 & 1.3 & 6.8 & 0.256 & 3.415\\
							&			&   $JEL_{w}$	 	 & 92.1 & 1.3 & 6.6 & 0.257 & 3.417\\
					\hline
\end{tabular}
\end{table}

\begin{ex}\label{ex2}\rm\;\
Consider the estimate of variance, where the kernel $h = \left ( x-y \right )^{2}/2$. And the U-statistic can be expressed as $U = \sum_{i=1}^{n} \sum_{j=1}^{n} \left ( y_{i} - y_{j} \right )^{2} / \left \{ n \left ( n-1 \right ) \right \}$.
\end{ex}

Table \ref{tab2} summarizes the results of example \ref{ex2}. It appears from Table \ref{tab2} that the coverage probability for the JEL, $JEL_d$ and $JEL_{w}$ methods are all quite closer to the nominal values than method NA. Although there is a small to moderate expansion in length. The only unsatisfactory aspect for the jackknife pseudo-empirical likelihood method is the tail error rates so generated are not balanced. It is evident that when increasing the sample size to n=150, the performance of the method proposed in this artical improve significantly.

\begin{table}[htbp]
\centering
\small 
\caption{Covergae probabilties, average lenths of the confidence intervals and tail error rates on the estimate of variance in Example \ref{ex2} when the norminal level is 0.95.}\vskip -1mm
\label{tab2}
\begin{tabular}{cccccccc}
					\hline
					\makebox[0.1\textwidth][c]{$\rho$}  & \makebox[0.1\textwidth][c]{n} & \makebox[0.1\textwidth][c]{CI} & \makebox[0.1\textwidth][c]{CP(\%)} & \makebox[0.1\textwidth][c]{L} & \makebox[0.1\textwidth][c]{U} & \makebox[0.1\textwidth][c]{AL} & \makebox[0.1\textwidth][c]{LB}\\
					\hline
					0.3	&	100	&	 		 NA		 	 & 90.4 & 0.4 & 9.2 & 5.705 & 7.813\\
							&			&			 JEL		 	 & 93.4 & 1.0 & 5.6 & 5.887 & 8.242 \\
							&			&   $JEL_{d}$  	 & 93.4 & 1.0 & 5.6 & 5.820 & 8.236 \\
							&			&   $JEL_{w}$  	 & 93.7 & 1.0 & 5.3 & 5.857 & 8.255\\
							&	150	&			 NA		 	 & 92.2 & 0.2 & 7.6 & 4.564 & 8.509\\
							&			&			 JEL		 	 & 94.8 & 0.5 & 4.7 & 4.663 & 8.788\\
							&			&   $JEL_{d}$  	 & 94.3 & 0.6 & 5.1 & 4.658 & 8.795\\
							&			&   $JEL_{w}$  	 & 94.3 & 0.7 & 5.0 & 4.662 & 8.809\\
					0.5	&	100	&	 		 NA		 	 & 91.1 & 0.5 & 8.4 & 2.080 & 2.858\\
							&			&			 JEL		 	 & 93.6 & 1.9 & 4.5 & 2.220 & 3.019\\
							&			&   $JEL_{d}$  	 & 93.2 & 1.8 & 5.0 & 2.336 & 3.017\\
							&			&   $JEL_{w}$  	 & 93.2 & 2.2 & 4.6 & 2.183 & 3.036\\
							&	150	&			 NA		 	 & 92.0 & 0.7 & 7.3 & 1.652 & 3.113 \\
							&			&			 JEL		 	 & 95.2 & 1.0 & 3.8 & 1.743 & 3.218\\
							&			&   $JEL_{d}$  	 & 94.9 & 0.9 & 4.2 & 1.781 & 3.225\\
							&			&   $JEL_{w}$	 	 & 95.5 & 1.0 & 3.5 & 1.733 & 3.239\\
					\hline
\end{tabular}
\end{table}

\section{Application to the China Household Finance Survey data}
China Household Finance Survey (CHFS) is a nationwide sample survey project, aiming to collect relevant information on the micro level of household finance. The CHFS describe household economic and financial behaviors comprehensively and meticulously. The data used in this section are from the China Household Finance Survey project (CHFS) managed by Survey and Research Center For China Household Finance of Southwestern University of Finance and Economics.

We apply proposed jackknife pseudo-empirical likelihood method to the 2019 CHFS data set, which covers 29 provinces, 170 cities, 345 districts and counties, and 1,360 village (residential) committees in China, with a sample size of 34,643 households. The data are representative at both the national and provincial levels. Note that the CHFS data sets are not publicly available but can be accessed through an approval process (https://chfs.swufe.edu.cn/). 

The response variable of interest that we chose is Y: resident happiness. The data set include a column of household sample weight $w_{i}$ for analytical purposes and we treat the sample as selected through a single-stage unequal probability sampling design with design weight $d_{i} = \pi_{i}^{-1}$ equal to the calibration weight $w_{i}$. We select the first, second and third tier city division variable as the auxiliary population information $x_{i}$. Due to the massive sample size, we adopt the method of equal probability random sampling to extract 10,000 samples for research.

We consider the variance of resident happiness, and the U-statistic can be represented as $U = \sum_{i=1}^{n} \sum_{j=1}^{n} \left ( y_{i} - y_{j} \right )^{2} / \left \{ n \left ( n-1 \right ) \right \}$, where n=10,000. The 95\% credible interval for jackknife empirical likelihood without auxiliary population information is (0.762,0.805). For comparison, the confidence intervals for jackknife empirical likelihood with basic design weights and jackknife empirical likelihood with calibration weights are both (0.758,0.800). From the result, we can conclude that the level of city has less effect on resident happiness.

\section{Conclusion}
In survey sampling, the performance of confidence intervals based on normal approximation is usually unsatisfactory. Therefore, in this article, we showcase the possibility of utilizing jackknife pseudo-empirical likelihood to solve problems involving the estimation of U-statistics in complex surveys, and at the estimation process reflect the characteristics of sampling design and the use of auxiliary information. The validity of such usage is examined through the use of theoretical justification and empirical evidence. Compared with the normal interval, the jackknife pseudo-empirical likelihood interval performs better in terms of balanced tail error rates and coverage probabilities.

Extending our proposed jackknife empirical likelihood method to other survey sampling subjects, such as stratified sampling, is currently under investigation. There have been significant advances in recent years on non-probability sampling. Exploring how to take full advantage of probability and non-probability sampling and combine them to better estimate statistics is an attractive topic. For such combined samples, it is valuable to adopt the method proposed for similar analysis.

\section*{Appendix}
In this section, we give proofs of main results.

\begin{proof}[Proof of Theorem \rm\ref{thm1}]

By rewriting $\tilde{d}_{i}\left ( s \right )\left ( v_{i}-\bar{V} \right )$ as $\tilde{d}_{i}\left ( s \right )\left ( v_{i}-\bar{V} \right )\left [ 1+\lambda\left ( \hat{v}_{i}-\bar{V} \right )-\lambda\left ( \hat{v}_{i}-\bar{V} \right ) \right ]$, we can rearrange (\ref{Sol1}) to obtain 
\beq \label{A1}
\lambda \sum_{i \in s}\frac{\tilde{d}_{i}\left ( s \right )\left ( v_{i}-\bar{V} \right )^2}{1+\lambda\left ( v_{i}-\bar{V} \right )}= \sum_{i \in s}\tilde{d}_{i}\left ( s \right ) \hat{v}_{i}-\bar{V}. \tag{A1}
\eeq
It follows from (\ref{A1}) that 
\beq \label{A2}
\frac{\left |\lambda  \right |}{1+\left |\lambda  \right |u^{\ast }}\sum_{i \in s} \tilde{d}_{i}\left ( s \right )\left ( \hat{v}_{i}-\bar{V} \right )^2\leq \left | \sum_{i \in s} \tilde{d}_{i}\left ( s \right )\hat{v}_{i}-\bar{V} \right | \tag{A2}
\eeq
where $u^{\ast }=max_{i \in s}\left | \hat{v}_{i}-\bar{V} \right |$ which is of order $o_{p} \left (n^{1/2} \right )$ by condition (Cl). It follows from Hajek (1960,1964), we can get the central limit theorem for a Horvitz-Thompson estimator, namely $\hat{\bar{V}}_{HT}=N^{-1} \sum_{i \in s} d_{i} \hat{v}_{i}$ of $\bar{V}$ is asymptotically normally distributed. Then, we can require $\hat{\bar{V}}_{HT}=\bar{V}+O_{p}\left( n^{-1/2} \right )$. Under conditions (C2), we have $\hat{N}/N=1+O_{p}\left( n^{-1/2} \right )$, where $\hat{N}=\sum_{i \in s} d_{i}$, which imply $\sum_{i \in s} \tilde{d}_{i}\left ( s \right )\hat{v}_{i}=\hat{\bar{V}}_{HT}/\left ( \hat{N}/N \right )= \bar{V}+O_{p}\left ( n^{-1/2} \right )$. Noting that $\sum_{i \in s} \tilde{d}_{i}\left ( s \right ) \left( \hat{v}_{i}-\bar{V} \right )^2$ is the Hajek-type estimator of $S_{V}^{2}$ which is of order $O \left ( 1 \right )$, it follows from (\ref{A2}) that we must have $\lambda = O_{p}\left ( n^{-1/2} \right )$ and, consequently, $max_{i \in s} \left | \lambda \left ( \hat{v}_{i}-\bar{V} \right ) \right |= o_{p}\left (1 \right )$. This together with (\ref{A1}) leads to 
\beq
\lambda =\left \{ \sum_{i \in s} \tilde{d}_{i}\left ( s \right ) \left ( \hat{v_{i}}-\bar{V} \right )^2 \right \}^{-1}\left ( \sum_{i \in s} \tilde{d}_{i}\left ( s \right )\hat{v_{i}}-\bar{V} \right )+o_{p}\left ( n^{-1/2} \right ). \nonumber
\eeq
Uing a Taylor series expansion of $log \left( 1+x \right )$ at $x=\lambda \left ( \hat{v}_{i}-\bar{V} \right)$ up to the second order, we obtain 
\beq
\begin{split}
\hat{r}_{JEL} \left ( \theta \right ) & = 2n^{\ast } \sum_{i \in s} \tilde{d}_{i}\left ( s \right )log \left \{ 1+\lambda\left ( \hat{v}_{i}-\bar{V} \right ) \right \} \\
& = n^{\ast } \left ( \sum_{i \in s} \tilde{d}_{i}\left ( s \right ) \hat{v}_{i}-\bar{V} \right )^{2}/\left ( \sum_{i \in s} \tilde{d}_{i}\left ( s \right ) \left (\hat{v}_{i}-\bar{V}  \right )^2 \right )+o_{p}\left ( 1 \right )
\end{split} \nonumber
\eeq
Since $\sum_{i \in s} \tilde{d}_{i} \left ( s\right )\left ( \hat{v}_{i}-\bar{V} \right )^{2}=S_{V}^{2}+o_{p} \left (1 \right )$, and $\sum_{i \in s} \tilde{d}_{i} \left ( s\right ) \hat{v}_{i}$ is asymptotically normal with mean $\bar{V}$ and variance $V_{p} \left( \hat{\bar{V}}_H \right )$ under Conditions (C2) and the central limit theorem for a Horvitz-Thompson estimator, the conclusion that the adjusted pseudo-empirical likelihood ratio statistic converges in distribution to $\chi _{1}^{2}$ follows immediately since $ \hat{r}_{JEL} \left ( \theta \right ) = \left \{ \sum_{i \in s} \tilde{d}_{i} \left (s \right ) \hat{v}_{i} -\bar{V} \right \}^{2} / V_{p}\left \{ \sum_{i \in s} \tilde{d}_{i} \left (s \right ) \hat{v}_{i} \right \}+o_{p}\left ( 1 \right )$.
\end{proof}

\begin{proof}[\bf\it Proof of Theorem \rm\ref{thm2}]

The arguments on the order of magnitude and the asymptotic expansion of the involved Lagrange multiplier are similar to those given in the proof of Theorem \ref {thm1}. There are two crucial arguments, however, which are unique to this proof. The $\tilde{p}_{i}$ which maximize $l_{JEL} \left( \theta \right )$ subject to $\sum_{i \in s}p_{i}=1$ and $\sum_{i \in s}p_{i}x_{i}=\bar{X}$ are given by $\tilde{p}_{i}=\tilde{d}_{i} \left ( s \right )/ \left \{ 1+\lambda^{T} \left ( x_{i}-\bar{X} \right ) \right \}$, where the $\lambda$ is the solution to 
\beq
\sum_{i \in s} \frac{\tilde{d}_{i} \left ( s \right ) \left ( \bf{x}_{i} - \bar{\bf{X}} \right ) }{1 + {\lambda}' \left ( \bf{x}_{i} - \bar{\bf{X}} \right )} = 0.  \nonumber 
\eeq
Under Conditions (C2), (C3) and $\hat{\bar{X}}_{HT}=N^{-1} \sum_{i \in s} d_{i} x_{i}$ is asymptotically normally distributed, we can show that $\left \| \lambda \right \|=O_{p} \left ( n^{-1/2} \right )$ and 
\beq
\lambda = \left \{ \sum_{i \in s} \tilde{d}_{i} \left ( s \right ) \left ({x_{i}-\bar{X}} \right ) \left ( {x_{i}-\bar{X}} \right )^{T} \right \}^{-1} \left ( \sum_{i \in s} \tilde{d}_{i} \left( s \right ) x_{i}-\bar{X} \right ) + o_{p} \left ( n^{-1/2} \right ) \nonumber
\eeq
With the term $n \sum_{i \in s} \tilde{d}_{i} \left ( s \right ) log \left ( \tilde{d}_{i} \left ( s \right ) \right )$ omitted, we obtain the following asymptotic expansion for $l_{ns} \left ( \tilde{p} \right )$: 
\beq \label{A3}
-\frac{n}{2} \left ( \sum_{i \in s} \tilde{d}_{i} x_{i} - \bar {X} \right )^{T}\left \{ \sum_{i \in s} \tilde{d}_{i} \left ( s \right )\left ({x_{i}-\bar{X}} \right )   \left ( {x_{i}-\bar{X}} \right )^{T} \right \}^{-1}\left ( \sum_{i \in s} \tilde{d}_{i} x_{i}- \bar {X} \right )+o_{p} \left ( 1 \right ). \tag{A3}
\eeq

To obtain a similar expansion for $l_{ns} \left ( \tilde{p} \left ( \bar{V} \right ) \right )$ where $\tilde{p} \left ( \bar{V} \right )$ maximize $l_{ns} \left ( p \right )$ subject to $\sum_{i \in s}p_{i}=1$, (\ref{Con1}) and (\ref{Con2}), our first crucial argument is to reformulate the constrained maximization problem as follows: let $r_{i}=\hat{v}_{i}- \bar{V}-B^{T} \left ( x_{i}-\bar{X} \right )$ where $B$ is defined by (\ref{B}). Then the set of constraints is equivalent to
\beq\label{A4}
\sum_{i \in s}p_{i}=1,\quad \sum_{i \in s}p_{i}x_{i}=\bar{X} \quad and \quad \sum_{i \in s}p_{i}r_{i} = 0. \tag{A4}
\eeq
With complete parallel development that leads to $l_{ns}$ given by (\ref{A3}), maximizing $l_{JEL}\left ( \tilde{p} \right )$ subject to (\ref{A4}) leads to the following expansion for $l_{JEL} \left ( \tilde{p} \left ( \bar{V} \right ) \right )$
\beq\label{A5}
-\frac{n}{2} \left ( \sum_{i \in s} \tilde{d}_{i} u_{i} - \bar {U} \right )^{T}\left \{ \sum_{i \in s} \tilde{d}_{i} \left ( s \right ) \left ( u_{i} - \bar {U} \right )  \left ( u_{i} -\bar {U} \right )^{T} \right \}^{-1}\left ( \sum_{i \in s} \tilde{d}_{i} u_{i} - \bar {U} \right )+o_{p} \left ( 1 \right ), \tag{A5}
\eeq
where $u_{i}=\left ( x_{i}^{T}, r_{i} \right )^{T}$ and $\bar{U}= \left (\bar{X}^{T},0 \right )^{T}$. Our second crucial argument is the observation that $\sum^{N}_{i=1} \left (x_{i}-\bar{X} \right )r_{i}$, i.e., the matrix involved in the middle of (\ref{A5}) is an estimate for its population counterpart which is block diagonal. It is straightforward to show that 
\beq
\tilde{r}_{JEL} \left ( \theta \right ) = -2 \left \{ l_{ns} \left ( \tilde{p} \left ( \bar{V} \right ) \right ) - l_{ns} \left ( \tilde{p} \right ) \right \} = n \left ( \sum_{i \in s} \tilde{d}_{i} \left (s \right ) r_{i} \right ) ^{2} / \left ( \frac{1}{N} \sum_{i=1}^{N} r_{i}^{2} \right ) +o_{p} \left ( 1 \right ). \nonumber
\eeq
The conclusion of the theorem follows since $\sum_{i \in s} \tilde{d}_{i} \left ( s \right ) r_{i}$ is asymptotically normal with mean 0 and variance $V_{p} \left \{ \sum_{i \in s} \tilde{d}_{i} \left ( s \right ) r_{i} \right \}$.
\end{proof}

\begin{proof}[Proof of Theorem \rm\ref{thm3}]

By rewriting $\tilde{w}_{i}\left ( s \right )\left ( v_{i}-\bar{V} \right )$ as $\tilde{w}_{i}\left ( s \right )\left ( v_{i}-\bar{V} \right )\left [ 1+\lambda\left ( \hat{v}_{i}-\bar{V} \right )-\lambda\left ( \hat{v}_{i}-\bar{V} \right ) \right ]$, we can rearrange (\ref{Sol3}) to obtain 
\beq \label{A6}
\lambda \sum_{i \in s}\frac{\tilde{w}_{i}\left ( s \right )\left ( v_{i}-\bar{V} \right )^2}{1+\lambda\left ( v_{i}-\bar{V} \right )}= \sum_{i \in s} \tilde{w}_{i} \left ( s \right ) \hat{v}_{i}- \bar{V}. \tag{A6}
\eeq
It follows from (\ref{A6}) that 
\beq \label{A7}
\frac{\left |\lambda  \right |}{1+\left |\lambda  \right |u^{\ast }}\sum_{i \in s} \tilde{w}_{i}\left ( s \right )\left ( \hat{v}_{i}-\bar{V} \right )^2\leq \left | \sum_{i \in s} \tilde{w}_{i}\left ( s \right )\hat{v}_{i}-\bar{V} \right | \tag{A7}
\eeq
where $u^{\ast }=max_{i \in s}\left | \hat{v}_{i}-\bar{V} \right |$ which is of order $o_{p} \left (n^{1/2} \right )$ by condition (Cl). It follows from Hajek (1960,1964), we can get the central limit theorem for a Horvitz-Thompson estimator, namely $\breve{\bar{V}}_{HT}=N^{-1} \sum_{i \in s} w_{i} \hat{v}_{i}$ of $\bar{V}$ is asymptotically normally distributed. Then, we can require $\breve{\bar{V}}_{HT}=\bar{V}+O_{p}\left( n^{-1/2} \right )$. Under conditions (C4), we have $\breve{N}/N=1+O_{p}\left( n^{-1/2} \right )$, where $\breve{N}=\sum_{i \in s} w_{i}$, which imply $\sum_{i \in s} \tilde{w}_{i}\left ( s \right )\hat{v}_{i}=\breve{\bar{V}}_{HT}/ \left ( \breve{N}/N \right )= \bar{V}+O_{p}\left ( n^{-1/2} \right )$. Noting that $\sum_{i \in s} \tilde{w}_{i}\left ( s \right ) \left( \hat{v}_{i}-\bar{V} \right )^2$ is the Hajek-type estimator of $S_{V}^{2}$ which is of order $O \left ( 1 \right )$, it follows from (\ref{A2}) that we must have $\lambda = O_{p}\left ( n^{-1/2} \right )$ and, consequently, $max_{i \in s} \left | \lambda \left ( \hat{v}_{i}-\bar{V} \right ) \right |= o_{p}\left (1 \right )$. This together with (\ref{A6}) leads to 
\beq
\lambda =\left \{ \sum_{i \in s} \tilde{w}_{i}\left ( s \right ) \left ( \hat{v_{i}}-\bar{V} \right )^2 \right \}^{-1}\left ( \sum_{i \in s} \tilde{w}_{i}\left ( s \right )\hat{v_{i}}-\bar{V} \right )+o_{p}\left ( n^{-1/2} \right ). \nonumber
\eeq
Uing a Taylor series expansion of $log \left( 1+x \right )$ at $x=\lambda \left ( \hat{v}_{i}-\bar{V} \right)$ up to the second order, we obtain 
\beq
\begin{split}
\breve{r}_{JEL}\left ( \theta \right ) & = 2 m \sum_{i \in s} \tilde{w}_{i}\left ( s \right )log \left \{ 1+\lambda\left ( \hat{v}_{i}-\bar{V} \right ) \right \} \\
& = m \left ( \sum_{i \in s} \tilde{w}_{i}\left ( s \right ) \hat{v}_{i}-\bar{V} \right )^{2}/\left ( \sum_{i \in s} \tilde{w}_{i}\left ( s \right ) \left (\hat{v}_{i}-\bar{V}  \right )^2 \right ) + o_{p}\left ( 1 \right )
\end{split} \nonumber
\eeq
Since $\sum_{i \in s} \tilde{w}_{i} \left ( s\right )\left ( \hat{v}_{i}-\bar{V} \right )^{2}=S_{V}^{2}+o_{p} \left (1 \right )$, and $\sum_{i \in s} \tilde{w}_{i} \left ( s\right ) \hat{v}_{i}$ is asymptotically normal with mean $\bar{V}$ and variance $V_{p} \left( \breve{\bar{V}}_H \right )$ under Conditions (C4) and the central limit theorem for a Horvitz-Thompson estimator, the conclusion that the adjusted pseudo-empirical likelihood ratio statistic converges in distribution to $\chi _{1}^{2}$ follows immediately since $ \breve{r}_{JEL}\left ( \theta \right ) =\left \{ \sum_{i \in s} \tilde{w}_{i} \left (s \right ) \hat{v}_{i} -\bar{V} \right \}^{2} / V_{p}\left \{ \sum_{i \in s} \tilde{w}_{i} \left (s \right ) \hat{v}_{i} \right \} + o_{p}\left ( 1 \right )$.
\end{proof}


\end{document}